\documentclass[a4paper]{article}
\usepackage{latexsym}
\usepackage{amsmath,amssymb,amsfonts}
\usepackage{graphicx} 
\renewcommand{\baselinestretch}{1.2}
\newtheorem{prethm}{{\bf Theorem}}

\newenvironment{thm}{\begin{prethm}{\hspace{-0.5
               em}{\bf.}}}{\end{prethm}}

\newtheorem{prepro}[prethm]{Proposition}

\newtheorem{prelem}[prethm]{Lemma}

\newtheorem{precor}[prethm]{Corollary}

\newtheorem{preremark}{{\bf Remark}}

\newenvironment{rem}{\begin{preremark}\em{\hspace{-0.5
              em}{\bf.}}}{\end{preremark}}

\newtheorem{preexample}{{\bf Example}}

\newtheorem{preque}{{\bf Question}}

\newtheorem{preproblem}{{\bf problem}}

\newtheorem{preproof}{{\bf Proof.}}

\newenvironment{proof}[1]{\begin{preproof}{\rm
               #1}\hfill{$\Box$}}{\end{preproof}}

\renewcommand{\thefootnote}

\begin{document}
\title{A Note on Signed $k$-Submatching in Graphs}
\author{S. Akbari$^{\,\rm a,c}$,~ M. Dalirrooyfard$^{\,\rm a,b}$,~ K. Ehsani$^{\,\rm b}$,~ R. Sherkati$^{\,\rm b}$\\
{\footnotesize {\em $^{\rm a}$Department of Mathematical Sciences, Sharif University of Technology,}}\\
{\footnotesize {\em $^{\rm b}$Department of Computer Engineering, Sharif University of Technology,}}\\
{\footnotesize {\em $^{\rm c}$School of Mathematics, Institute for Research in Fundamental Sciences (IPM),}}\\
{\footnotesize {\em P.O. Box 19395-5746, Tehran, Iran}}}
\footnotetext{E-mail Addresses: {\tt s\_akbari@sharif.edu}, {\tt mdalir\_rf@ee.sharif.edu}, {\tt kehsani@ce.sharif.edu}, {\tt sherkati@ce.sharif.edu}}
\date{}
\maketitle


\begin{abstract}

 Let $G$ be a graph of order $n$. For every $v\in V(G)$, let $E_G(v)$ denote the set of all edges incident with $v$. A signed $k$-submatching of $G$ is a function $f:E(G)\longrightarrow \{-1,1\}$, satisfying $f(E_G(v))\leq 1$ for at least $k$ vertices, where $f(S)=\sum_{e\in S}f(e)$, for each $ S\subseteq E(G)$. The maximum of the value of $f(E(G))$, taken over all signed $k$-submatching $f$ of $G$, is called the signed $k$-submatching number and is denoted by $\beta ^k_S(G)$. In this paper, we prove that for every graph $G$ of order $n$ and for any positive integer $k \leq n$, $\beta ^k_S (G) \geq n-k - \omega(G)$, where $w(G)$ is the number of components of $G$. This settles a conjecture proposed by Wang. Also, we present a formula for the computation of $\beta_S^n(G)$.
\vspace{1mm} {\renewcommand{\baselinestretch}{1}
\parskip = 0 mm

}


\end{abstract}

\section{Introduction}
 \hspace{0.5cm} Let $G$ be a simple graph with the vertex set $V(G)$ and edge set $E(G)$. For every $v\in V(G)$, let $N(v)$ and $E_G(v)$ denote the set of all
 neighbors of $v$ and the set of all edges incident with $v$, respectively. A {\it signed $k$-submatching} of a graph $G$ is a function $f:E(G)\longrightarrow \{-1,1\}$, satisfying $f(E_G(v))\leq 1$ for at least $k$ vertices, where $f(S)=\sum_{e\in S}f(e)$, for each $ S\subseteq E(G)$. The maximum value of $f(E(G))$, taken over all signed $k$-submatching $f$, is called the {\it signed $k$-submatching number} of $G$ and is denoted by $\beta ^k_S(G)$. We refer to signed $n$-submatching as {\it signed submatching}. The concept of signed matching has been studied by several authors, for instance see \cite{algo}, \cite{edge_ghamesh}, \cite{wang_cover} and \cite{wang_signed}.

Throughout this paper, changing $f(e)$ to $-f(e)$ for an edge $e$ is called {\it switching the value} of $e$.
Let $T$ be a trail with the edges $e_1,\ldots,e_m$ and $f$ be a signed submatching of $G$. Call $T$ a {\it good trail}, if $f(e_i) = - f(e_{i+1})$ for $i=1,\ldots,m - 1$. If $f(e_1)=a$ and $f(e_m)=b$, then we call $T$ a {\it good $(a,b)$-trail}. Define $O_f(G) = \{v\in V(G) \, | \, d(v)\equiv1\mbox{(mod 2)} ,\,\, f(E_G(v))<1 \}$. A vertex is called {\it odd} if its degree is odd.
The following conjecture was proposed in \cite{conj}.

 \noindent {\bf Conjecture.} Let $G$ be a graph without isolated vertices. Then for any positive integer $k$, $$
\beta_S^k(G)\ge n-k-\omega(G),
$$
where $\omega(G)$ denotes the number of components of $G$.

In this note we prove this conjecture. Before stating the proof, we need the following result.

\begin{thm}
\label{connect_conj}
Let $G$ be a connected graph of order $n$. Then for any positive integer $k\leq n$, $\beta^{k}_S(G)\geq n-k-1.$
\end{thm}

\begin{proof}
{ If $G$ is a cycle, then by Theorem $2$ of \cite{conj} the assertion is obvious. Thus assume that $G$ is not a cycle. Now, we apply induction on $ |E(G)| - |V(G)|$. Since $ G $ is connected, $ |E(G)| - |V(G)| \ge -1$. If $ |E(G)| - |V(G)| = -1$, then $ G $ is a tree and so by Theorem $6 $ of \cite{conj}, we are done. Now, suppose that the assertion holds for every graph $H$ with $ |E(H)| - |V(H)| \leq t \, (t\geq -1)$ and $ G $ be a connected graph such that $ |E(G)| - |V(G)| = t + 1$. Since $ |E(G)|- |V(G)| \ge 0$, $ G $ contains a cycle $ C$ and there exists a vertex $v$ such that $ v \in V(C) $ and $ d(v) \ge 3$. Assume that $ u,w \in N(v) \cap V(C)$. Let $ x \in N(v)\backslash \{ u,w \}$. Remove two edges $ vw $ and $xv$ and add a new vertex $ v'$. Join $ v' $ to both $ x$ and $w$. Call the new graph by $ G'$. Clearly, $G'$ is connected and $|E(G')| - |V(G')| = t$. By the induction hypothesis, $\beta^{k+1}_S(G')\geq |V(G')|-k-2=n-k-1$. We claim that $\beta^{k}_S(G)\geq \beta^{k+1}_S(G')$. Let $f$ be a signed $(k+1)$-submatching of $G'$ such that $f(E(G'))=\beta^{k+1}_S(G')$. Define a function $g$ on $E(G)$ as follows:

 For every $e\in E(G)\setminus \{vx, vw\}$, let $g(e)=f(e)$. Moreover, define $g(xv)=f(xv')$ and $g(vw)=f(v'w)$.
 It is not hard to see that $g$ is a $k$-submatching of $G$. So $\beta^{k}_S(G)\geq g(E(G))=\beta^{k+1}_S(G')$. Thus $\beta^{k}_S(G)\geq n-k-1$, and the claim is proved. The proof is complete.}
\end{proof}

Now, using the previous theorem we show that the conjecture holds.

\begin{thm}
Let $G$ be a graph of order $n$ without isolated vertices. Then for any positive integer $k\le n$,
$$
\beta_S^k(G)\ge n-k-\omega(G),
$$
where $\omega(G)$ denotes the number of components of $G$.
\end{thm}

\begin{proof}
{
For the abbreviation let $\omega=\omega(G)$. If $\omega = 1$, then by Theorem \ref{connect_conj} the assertion holds. Now, suppose that $\omega>1$ and $G_1,\ldots,G_\omega$ are all components of $G$. Let $f:E(G)\longrightarrow \{-1,1\}$, be a signed $k$-submatching function such that $f(E(G))=\beta_S^k(G)$. Suppose that $A\subset \{v\in V(G)\,|\,f(E_G(v))\le1\}$ and $|A|=k$. Let $k_i = |\{v\in V(G_i)\cap A\,|\,f(E_G(v))\le1\}|$, for $i=1,\ldots,\omega$. Obviously, $\sum_{i=1}^\omega k_i = k$. By Theorem \ref{connect_conj}, $\beta^{k_i}_S(G_i)\ge |V(G_i)| - k_i-1$, for $i=1,\ldots,\omega$. Now, we show that $\beta_S^{k_i}(G_i)=f(E(G_i))$. By contradiction, suppose that $\beta_S^{k_i}(G_i)>f(E(G_i))$, for some $i$, $i=1,\ldots,\omega$. Let $g:E(G)\longrightarrow \{-1,+1\}$ be a function such that $g(e)=f(e)$, for every $e\in E(G)\setminus E(G_i)$ and the restriction of $g$ on $E(G_i)$ is a signed $k_i$-submatching with $g(E(G_i))=\beta_S^{k_i}(G_i)$. So we conclude that $g(E(G))>\beta_S^k(G)$, a contradiction. Thus $\beta_{S}^k(G)=f(E(G))=\sum_{i=1}^{\omega}f(E(G_i))=\sum_{i=1}^\omega\beta^{k_i}_S(G_i)\ge \sum_{i=1}^{\omega}(|V(G_i)| - k_i-1)=|V(G)|-k-\omega$.
}
\end{proof}

Now, suppose that $G$ is a connected graph containing exactly $2k$ odd vertices. 
Let $P$ be a partition of the edge set into $m$ trails, say $T_1, \ldots, T_m$, for some $m$. Call $P$ a {\it complete partition} if $m=k$. By Theorem 1.2.33 of \cite{west}, for every connected graph with $2k$ odd vertices there exists at least one complete partition. Note that for every odd vertex $v\in V(G)$, there exists $i$ such that $v$ is an end point of $T_i$, where $P: T_1,\ldots, T_k$ is a complete partition of $G$. So we obtain that the end vertices of $T_i$ are odd and they are mutually disjoint, for $i=1,\ldots, k$.
Now, define $\tau(P)=|\{i\,\,|\,\,|E(T_i)|\equiv1(\mbox{mod 2})\}|$. Let $\eta(G)=\max{\tau(P)}$, taken over all complete partitions of $G$. In the next theorem we provide an explicit formula for the signed $n$-submatching number of a graph.

\begin{thm}
For every non-Eulerian connected graph $G$ of order $n$, $\beta_S^n(G)=\eta(G)$.
\end{thm}

\begin{proof}
{

For the simplicity, let $O_f=O_f(G)$. Let $f$ be a signed submatching such that $|O_f|=max(|O_g|)$ taken over all signed submatching $g$ with $g(E(G))=\beta_S^n(G)$. We prove that $f(E_G(v))\ge -1$, for every $v\in V(G)$.

By contradiction suppose that there is a vertex $v\in V(G)$ such that $f(E_G(v))\le -2$. Let $W$ be a longest good $(-1,\pm1)$-trail starting at $v$. Suppose that $W$ ends at $u$. There are two cases: 

{\bf Case 1.} Assume that $u\neq v$. If $W$ is a good $(-1,-1)$-trail, then $f(E_G(u))\le -1$, since otherwise there exists $e\in E_G(u)\setminus E(W)$ such that $f(e)=1$, therefore $W$ can be extended and it contradicts the maximality of $|E(W)|$. Now, switch the values of all edges of $W$ to obtain a function $g$ on $E(G)$, where $g(E_G(x))=f(E_G(x))$ for every $x\in V(G)\setminus \{u,v\}$, and $g(E_G(x))=f(E_G(x))+2$ for $x\in \{u,v\}$. Thus $g$ is a signed submatching of $G$ such that $g(E(G))= \beta^n_S(G) + 2$, a contradiction.\\
 If $W$ is a good $(-1,1)$-trail, then $f(E_G(u))=1$, since otherwise there exists $e\in E_G(u)\setminus E(W)$, where $f(e)=-1$, a contradiction. Now, switch the values of all the edges of $W$ to obtain a function $g$ on $E(G)$, where $g(E_G(x))=f(E_G(x))$ for $x\in V(G)\setminus \{u,v\}$, $g(E_G(u))=-1$ and $g(E_G(v))< 1$. So $g$ is a signed submatching of $G$ such that $g(E(G))=\beta_S^n(G)$ and $|O_g|=|O_f|+1$, a contradiction. 

{\bf Case 2.} Now, let $u=v$. Note that $W$ is a good $(-1,-1)$-trail, since otherwise $\sum_{e \in E(W) \cap N_G(v)} f(e) = 0$ and using the inequality $f(E_G(v))\le -2$, we conclude that there exists $e\in E_G(v)\setminus E(W)$, such that $f(e)=-1$. Therefore $W$ can be extended, a contradiction.

If $f(E_G(v))\le -3$, then switch the values of all edges of $W$ to obtain a signed submatching $g$ such that $g(E(G))= \beta^n_S(G) + 2$, a contradiction. Now, assume that $f(E_G(v))=-2$. We show that $f(E_G(t)) = 0$, for every $t \in V(W)\setminus \{v\}$. By contradiction, suppose that there exists $x \in V(W)\setminus \{v\}$, such that $f(E_G(x)) \neq 0$. Let $e_1,\ldots, e_m$ be all edges of $W$. Assume that $e_i$ and $e_{i+1}$ are two consecutive edges of $W$ which are incident with $x$. With no loss of generality, assume that $f(e_i)=-1$. First, suppose that $f(E_G(x))\le -1$. Call the sub-trail induced on the edges $e_1,e_2,\ldots, e_i$ by $W_1$. Clearly, $W_1$ is a good $(-1,-1)$-trail. Switch the values of all edges of $W_1$ to obtain a signed submatching $g$ such that $g(E(G))= \beta^n_S(G) + 2$, a contradiction. 
Next, suppose that $f(E_G(x))=1$. Call the sub-trail induced on the edges $e_{i+1},\ldots, e_m$ by $W_2$. Clearly, $W_2$ is a good $(1,-1)$-trail. Switch the values of all edges of $W_2$ to obtain a signed submatching $g$ such that $g(E_G(x))=-1$, $g(E_G(v))=0$ and $g(E_G(z))=f(E_G(z))$, for every $z\in E(G) \setminus \{x,v\}$. So $g(E(G))=\beta_S^n(G)$ and $|O_g|=|O_f|+1$, a contradiction. Thus, $f(E_G(t)) = 0$, for every $t \in V(W)\setminus \{v\}$.

Now, we show that $E_G(v)\subseteq E(W)$. By contradiction assume that there exists $e\in E_G(v)\setminus E(W)$. If $f(e)=1$, then $W$ can be extended, a contradiction. If $f(e)=-1$, then $f(E_G(v))\le -3$ which contradicts $f(E_G(v))= -2$. Thus $E_G(v)\subseteq E(W)$.
Since $G$ is non-Eulerian, there are $x\in V(W)\setminus \{v\}$ and $y\in V(G)$ such that $xy\notin E(W)$. Let $W'$ be a longest good trail in $G\setminus E(W)$ whose first vertex and first edge are $x$ and $xy$, respectively. Suppose that $W'$ ends at $y'$ and the last edge of $W'$ is $e$. We have two possibilities:

If $y'=x$, then we show that $W'$ is a good $(1,-1)$ or $(-1,1)$-trail. To see this, since $f(E_{G}(x))=0$, we obtain that $f(E_G(x)\setminus E(W))=0$. If $f(e)=f(xy)$, then there exists $e'\in E_G(x)\setminus (E(W)\cup E(W'))$ such that $f(e')=-f(xy)$. So $W'$ can be extended, a contradiction. Thus $f(e) \neq f(xy)$. It is not hard to see that the trail with the edges $E(W)\cup E(W')$ is a good $(-1,-1)$-trail starting at $v$, a contradiction.

Now, suppose that $y'\neq x$. Assume that $x$ is the common end point of $e_j$ and $e_{j+1}$, for some $j,$ $1\leq j\leq m-1$. With no loss of generality assume that $f(e_j)=-f(xy)$. Consider the trail $W'' : e_1,\ldots, e_j, W'$.  Since $E_G(v)\subseteq E(W)$, $y' \neq v$. If $y' \in V(W)$, then $f(E_G(y')) = 0$ and $\sum_{z\in (E_G(W)\cup E_G(W'))\cap E_G(y')} f(z) = f(e)$. Hence there exists $e' \in E_G(y')\setminus (E_G(W) \cup E_G(W'))$ such that $f(e') = -f(e)$, which contradicts the maximality of $|E(W')|$. Thus $y' \notin V(W)$ and so $W''$ is a maximal good trail in $G$. So we reach to the Case $1$ which we discussed before (Note that in the Case $1$ we used just the maximality of the length of $W$). So we proved that $f(E_G(z))\ge -1$, for every $z\in V(G)$. 

In the sequel assume that $G$ has exactly $2k$ odd vertices. We would like to partition $G$ into $k$ good trails.

Let $T:e_1,\ldots,e_m$ be a longest good trail in $G$. Suppose that $T$ starts at $u_1$ and ends at $u_2$, where $u_1,u_2\in V(G)$. First, we show that $u_1 \neq u_2$. By contradiction assume that $u_1=u_2$. Suppose that $f(e_1) \neq f(e_m)$. Since $G$ is non-Eulerian, there exists $e\in E(G) \setminus E(T)$ and $e$ is adjacent to the common end point of $e_i$ and $e_{i+1}$ for some $i,\, i=1,\ldots,m$ ($e_{m+1}=e_1$). With no loss of generality assume that $f(e) \neq f(e_i)$, so $T':e, e_i, e_{i-1},\ldots,e_1,e_m,\ldots,e_{i+1}$ is a good trail with $m+1$ edges, a contradiction. Now, suppose that $f(e_1)=f(e_m)$. Since $\sum_{z\in E_G(u_1)\cap E(T)} f(z)=2f(e_1)$ and $|f(E_G(u_1))|\leq 1$, we obtain that there exists $a\in E_G(u_1)\setminus E(T)$ such that $f(a)\neq f(e_1)$. So $T$ can be extended, a contradiction.\\   
Hence $u_1\neq u_2$. Since $f(E_G(v))=0$, for every $v\in V(G)$ of even degree, we obtain that $u_1$ and $u_2$ have odd degrees. Indeed, if $u_1$ has even degree, then $f(E_G(u_1))=0$ and so $T$ can be extended, a contradiction. Now, we show that $E_G(u_1)\cup E_G(u_2) \subseteq E(T)$. By contradiction, suppose that there is an edge $e\in E_G(u_1)\setminus E(T)$. Clearly, $f(e)=f(e_1)$. Since $\sum_{a\in E_G(u_1)\cap E(T)}f(a) =f(e_1)$, it is not hard to see that $|f(E_G(u_1))|\ge 2$, a contradiction. Hence $E_G(u_1)\cup E_G(u_2) \subseteq E(T)$. 

Let $G'=G\setminus (E(T)\cup \{u_1,u_2\})$. First, we prove that $G'$ has no Eulerian component. By contradiction suppose that $H$ is an Eulerian component of $G'$. Since $|f(E_G(v))|\le 1$, for every $v\in V(G)$, we have $f(E_G(v))=0$, for every $v\in V(H)$. It is straight forward to see that there is an Eulerian circuit $C:t_1,t_2,\ldots,t_{|E(H)|}$ of $H$ such that $f(t_i)=-f(t_{i+1})$, for $i=1,\ldots , |E(H)|-1$. Clearly, $|E(C)|\equiv \sum_{e\in E(C)}f(e) \equiv \frac{\sum_{v\in V(C)}f(E_H(v))}{2}\equiv 0\,(\mbox{mod }\,2)$. Hence, $f(t_1)=-f(t_{|E(H)|})$. Since $G$ is connected and all of the edges of $u_1$ and $u_2$ belong to $E(T)$, there exists $v\in V(H) \cap V(T)$. It is not hard to see that we have a good trail with the edge set $E(T)\cup E(C)$ which is longer than $T$, a contradiction. So if $k=2$, then $E(G')= \emptyset$, and $E(G)$ forms a good trail. Now, apply induction on $k$. Suppose that $k>2$. Let $H_1, \ldots, H_r$ be all components of $G'$, where $H_i$ has $k_i$ odd vertices ($k_i\ge 2$), for $i=1,\ldots , r$. It is clear that $f$ is a signed submatching of $H_i$ such that $f(E(H_i))=\beta^{|V(H_i)|}_S(H_i)$ and $O_f(H_i)=\max O_g(H_i)$ taken over all signed submatching $g$ with $g(E(H_i))=\beta^{|V(H_i)|}_S(H_i)$. So $E(H_i)$ can be decomposed into $k_i$ good trails, for $i=1,\ldots, r$. Hence, $G$ has a complete partition, say $P$, into $k$ good trails. Obviously, $f(E(G))\leq \tau(P) \leq \eta(G)$. Thus, $\beta^n_S(G)\leq \eta(G)$.
Now, we give a signed submatching $f$ such that $f(E(G))=\eta(G)$.

Consider a complete partition $P$ of the edge set of $G$, where $\tau(P)=\eta(G)$. For each trail $T_i$ assign $+1$ and $-1$ to the edges of $T_i$, alternatively, to obtain a signed submatching $f$ where $f(E(G))=\eta(G)$. So the proof is complete.}
\end{proof}

\begin{rem}
For every Eulerian graph $G$ of size $m$, $ \beta^n_S(G)=0$ if $m$ is even and $\beta^n_S(G)=-1$ if $m$ is odd. To see this, let $f$ be a signed submatching of $G$ such that $f(E(G))=\beta_{S}^n(G)$. Since  the degree of each vertex of $G$ is even, $f(E_G(v))\le 0$, for every $v\in V(G)$. Thus $f(E(G)) = \frac{1}{2}\sum_{v\in V(G)} f(E_G(v))\le 0$. Therefore, $ \beta^n_S(G)\le 0$, if $m$ is even and $\beta^n_S(G)\le -1$, if $m$ is odd.
Now, consider an Eulerian circuit of $G$. Assign $-1$ and $+1$ to the edges of this Eulerian circuit, alternatively to obtain a signed submatching $g$ with the desired property.
\end{rem}

\end{document}